%% file: Diagonalization_QPL_final_proceeding_submitted.tex
\documentclass[british,copyright, creativecommons]{eptcs}
\usepackage[T1]{fontenc}
\usepackage[latin9]{inputenc}
\usepackage{amsthm}
\usepackage{amsmath}
\usepackage{amssymb}
\usepackage{breakurl}

\makeatletter
  \theoremstyle{plain}
  \newtheorem{assumption}{\protect\assumptionname}
  \theoremstyle{definition}
  \newtheorem{defn}{\protect\definitionname}
  \theoremstyle{plain}
  \newtheorem{ax}{\protect\axiomname}
  \theoremstyle{plain}
  \newtheorem{prop}{\protect\propositionname}
\theoremstyle{plain}
\newtheorem{thm}{\protect\theoremname}
  \theoremstyle{plain}
  \newtheorem{cor}{\protect\corollaryname}
  \theoremstyle{plain}
  \newtheorem{lem}{\protect\lemmaname}

\providecommand{\event}{QPL 2015}
\def\titlerunning{Operational axioms for diagonalizing states}
\def\authorrunning{ G. Chiribella \& C. M. Scandolo}

\makeatother
\input{myQcircuit}
\makeatletter

\newcommand{\rA}{\mathrm{A}}
\newcommand{\rB}{\mathrm{B}}
\newcommand{\rC}{\mathrm{C}}

\newcommand{\cA}{\mathcal{A}}
\newcommand{\cB}{\mathcal{B}}
\newcommand{\cC}{\mathcal{C}}
\newcommand{\cU}{\mathcal{U}}

\makeatother

\usepackage{babel}
  \providecommand{\assumptionname}{Assumption}
  \providecommand{\axiomname}{Axiom}
  \providecommand{\definitionname}{Definition}
  \providecommand{\lemmaname}{Lemma}
  \providecommand{\propositionname}{Proposition}
\providecommand{\corollaryname}{Corollary}
\providecommand{\theoremname}{Theorem}

\begin{document}

\title{Operational axioms for diagonalizing states}

\author{Giulio Chiribella\email{giulio@cs.hku.hk} \institute{Department
of Computer Science, University of Hong Kong, Hong Kong} \and Carlo
Maria Scandolo\email{carlomaria.scandolo@st-annes.ox.ac.uk} \institute{Department
of Computer Science, University of Oxford, Oxford, UK}}
\maketitle
\begin{abstract}
In quantum theory every state can be diagonalized, i.e.\ decomposed
as a convex combination of perfectly distinguishable pure states.
This elementary structure plays an ubiquitous role in quantum mechanics,
quantum information theory, and quantum statistical mechanics, where
it provides the foundation for the notions of majorization and entropy.
A natural question then arises: can we reconstruct these notions from
purely operational axioms? We address this question in the framework
of general probabilistic theories, presenting a set of axioms that
guarantee that every state can be diagonalized. The first axiom is
Causality, which ensures that the marginal of a bipartite state is
well defined. Then, Purity Preservation states that the set of pure
transformations is closed under composition. The third axiom is Purification,
which allows to assign a pure state to the composition of a system
with its environment. Finally, we introduce the axiom of Pure Sharpness,
stating that for every system there exists at least one pure effect
occurring with unit probability on some state. For theories satisfying
our four axioms, we show a constructive algorithm for diagonalizing
every given state. The diagonalization result allows us to formulate
a majorization criterion that captures the convertibility of states
in the operational resource theory of purity, where random reversible
transformations are regarded as free operations.
\end{abstract}

\section{Introduction}

A canonical route to the foundations of quantum thermodynamics is
provided by the theory of majorization, used to define an ordering
among states according to their degree of mixedness \cite{Uhlmann1,Uhlmann2,Uhlmann3,Thirring}.
In recent years, the applications of majorization have seen remarkable
developments in the study of quantum and nano thermodynamics \cite{Faist-Landauer,Horodecki-Oppenheim-2,Nicole,Brandao}.
The viability of this approach relies heavily on the Hilbert space
framework, for it is based on the fact that density operators can
be diagonalized. Ideally, however, it would be desirable to have an
axiomatic foundation of quantum thermodynamics based on purely operational
axioms.

The problem can be addressed in the framework of general probabilistic
theories \cite{Hardy-informational-1,D'Ariano,Barrett,Barnum-1,Chiribella-purification,Chiribella-informational,Barnum-2,Hardy-informational-2,hardy2011,Chiribella14}.
The first step in this direction is to consider probabilistic theories
that satisfy an operational version of the spectral theorem, according
to which every state can be ``diagonalized'', i.e.\ decomposed
as a mixture of perfectly distinguishable pure states. At this point
there are two options: one option is to demand the diagonalizability
of states as an axiom. This approach has been adopted in Refs.~\cite{Barnum-interference,Krumm-thesis,Krumm-Muller},
also in relation to the issue of defining majorization in general
probabilistic theories. The other option is to reduce diagonalization
to other operational axioms, which may provide deeper insights on
the conceptual foundations of quantum thermodynamics. This approach
will be the subject of the present paper.

A diagonalization result from operational principles was proved by
D'Ariano, Perinotti, and one of the authors in the context of the
axiomatization of quantum theory in Ref.~\cite{Chiribella-informational}
(hereafter referred to as CDP), although the proof therein used the
full set of axioms implying quantum theory. In this paper we derive
the diagonalizability of states from a strictly weaker set of axioms,
which is compatible with quantum theory on real Hilbert spaces and,
with other potential generalizations of quantum theory, such as the
fermionic theory recently proposed by D'Ariano \emph{et al} in Refs.~\cite{Fermionic1,Fermionic2}.
Our list of axioms consists of:
\begin{itemize}
\item two of the six CDP axioms (Causality and Purification);
\item one axiom (Purity Preservation) that is close to the CDP axiom Atomicity
of Composition, although not exactly equivalent to it;
\item a new axiom, which we name \emph{Pure Sharpness}.
\end{itemize}
Pure Sharpness stipulates that every physical system has at least
one pure effect occurring with unit probability on some state. Such
a pure effect can be seen as part of a yes-no test designed to check
an elementary property, in the sense of Piron \cite{PironBook}. In
these terms, Pure Sharpness requires that for every system there exist
at least one property, and at least one state possessing such a property.
Note that none of our axioms assumes that perfectly distinguishable
states exist. A priori, the general probabilistic theories considered
here may not contain \emph{any} pair of perfectly distinguishable
states---operationally, this would mean that no system described by
the theory could be used to transmit a classical bit with zero error.
The existence of perfectly distinguishable states, and the fact that
every state can be broken down into a mixture of perfectly distinguishable
pure states are non-trivial consequences of the axioms.

Note that the presence of Purification among the axioms excludes from
the start the case of classical probability theory. Indeed, the aim
of our work is \emph{not} to provide the most general conditions for
the diagonalization of states, but rather to derive diagonalization
as a first step towards an axiomatic foundation of quantum thermodynamics.
In particular, we are searching for axioms that capture the characteristic
traits of quantum thermodynamics, such as the link with the resource
theory of entanglement \cite{Chiribella-Scandolo15-1}. From this
point of view, Purification is an almost mandatory choice, in that
it sets up a fundamental relation between mixed states and pure entangled
states. More importantly, Purification is deeply related to the thermodynamic
procedure that consists in considering the system in interaction with
its environment in such a way that the composite system is isolated.
In this scenario, Purification guarantees that one can always associate
a pure state with the composite system and that the overall evolution
of system and environment can be treated as reversible. In this way,
thermodynamics is reconciled with the paradigm of reversible dynamics
at the fundamental level. In the concrete Hilbert space setting, the
purified view of quantum thermodynamics has been adopted in a number
of works aimed at deriving the microcanonical and canonical ensembles
\cite{Bocchieri,Seth-Lloyd,Lubkin,Gemmer-Otte-Mahler,Canonical-typicality,Popescu-Short-Winter,Mahler-book,Concentration-measure,BrandaoQIP2015},
an idea that has been recently explored also in general probabilistic
theories \cite{Dahlsten,Muller-blackhole}.

After deriving the diagonalizability of states, we discuss the implications
of the result. In particular, we discuss the relation of majorization,
defined in terms of the probability distributions arising from diagonalization.
Combining our axioms with an additional axiom, known as Strong Symmetry
\cite{Barnum-interference}, we then show that majorization completely
determines the convertibility of states in the operational resource
theory of purity \cite{Chiribella-Scandolo15-1}, where random reversible
transformations are viewed as free operations. It remains as an open
question whether in the context of our axioms Strong Symmetry can
be replaced with a weaker requirement \cite{Majorization}.

The paper is structured as follows: in section~\ref{sec:Framework}
we introduce the basic framework. The four axioms for diagonalization
are presented in section~\ref{sec:Axioms}, and their consequences
are examined in section~\ref{sec:Consequences-of-the}. Section~\ref{sec:Diagonalization-of-states}
contains the main result, namely the diagonalization theorem. In section~\ref{sec:Combining}
we discuss a number of results that arise from the combination of
diagonalization with the Strong Symmetry axiom. Using these results,
section~\ref{sec:puriresource} analyses majorization and its applications
to the resource theory of purity. The conclusions are drawn in section~\ref{sec:Conclusions}.

\section{Framework\label{sec:Framework}}

The present analysis is carried out in the framework of general probabilistic
theories, adopting the specific variant of Refs.~\cite{Chiribella-purification,Chiribella-informational,Chiribella14},
known as the framework of \emph{operational-probabilistic theories}
(\emph{OPTs}). OPTs arise from the marriage of the graphical language
of symmetric monoidal categories \cite{Abramsky2004,cqm,Coecke-Picturalism,Categories-physicist,Selinger}
with the toolbox of probability theory. Here we give a quick summary
of the framework, referring the reader to the original papers and
to the related work by Hardy \cite{Hardy-informational-2,hardy2013}
for a more in-depth presentation. A comprehensive review of the OPT
framework is presented in the book chapter \cite{QuantumFromPrinciples}.

Physical processes can be combined in sequence or in parallel, giving
rise to circuits like the following\[
\begin{aligned}\Qcircuit @C=1em @R=.7em @!R { & \multiprepareC{1}{\rho} & \qw \poloFantasmaCn{\rA} & \gate{\cA} & \qw \poloFantasmaCn{\rA'} & \gate{\cA'} & \qw \poloFantasmaCn{\rA''} &\measureD{a} \\ & \pureghost{\rho} & \qw \poloFantasmaCn{\rB} & \gate{\cB} & \qw \poloFantasmaCn{\rB'} &\qw &\qw &\measureD{b} }\end{aligned}~.
\]Here, $\mathrm{A}$, $\mathrm{A}'$, $\mathrm{A}''$, $\mathrm{B}$,
$\mathrm{B}'$ are \emph{systems}, $\rho$ is a bipartite \emph{state},
$\mathcal{A}$, $\mathcal{A}'$ and $\mathcal{B}$ are \emph{transformations},
$a$ and $b$ are \emph{effects}. Circuits with no external wires,
like the one in the above example, are associated with probabilities.
We denote by 
\begin{itemize}
\item $\mathsf{St}\left(\mathrm{A}\right)$ the set of states of system
$\mathrm{A}$ 
\item $\mathsf{Eff}\left(\mathrm{A}\right)$ the set of effects on $\mathrm{A}$ 
\item $\mathsf{Transf}\left(\mathrm{A},\mathrm{B}\right)$ the set of transformations
from $\mathrm{A}$ to $\mathrm{B}$ 
\item $\mathrm{A}\otimes\mathrm{B}$ the composition of systems $\mathrm{A}$
and $\mathrm{B}$. 
\item $\mathcal{A}\otimes\mathcal{B}$ the parallel composition of the transformations
$\mathcal{A}$ and $\mathcal{B}$. 
\end{itemize}
A particular system is the trivial system $\mathrm{I}$ (mathematically,
the unit of the tensor product), corresponding to the degrees of freedom
ignored by the theory. States (resp.\ effects) are transformations
with the trivial system as input (resp.\ output). We will often make
use of the short-hand notation $\left(a|\rho\right)$ to denote the
scalar\[
\left(a|\rho\right)~:=\!\!\!\!\begin{aligned}\Qcircuit @C=1em @R=.7em @!R { & \prepareC{\rho}    & \qw \poloFantasmaCn{\rA}  &\measureD{a}}\end{aligned}~,
\]and of the notation $\left(a\right|\mathcal{C}\left|\rho\right)$
to mean\[
\left(a\right|\cC\left|\rho\right)~:=\!\!\!\!\begin{aligned}\Qcircuit @C=1em @R=.7em @!R { & \prepareC{\rho}    & \qw \poloFantasmaCn{\rA}  &\gate{\cC}  &\qw \poloFantasmaCn{\rB} &\measureD{a}}\end{aligned}~.
\]We identify the scalar $\left(a|\rho\right)$ with a real number in
the interval $\left[0,1\right]$, representing the probability of
a joint occurrence of the state $\rho$ and the effect $a$ in a circuit
where suitable non-deterministic elements are put in place. The fact
that scalars are real numbers induces a notion of sum for transformations,
whereby the sets $\mathsf{St}\left(\mathrm{A}\right)$, $\mathsf{Transf}\left(\mathrm{A},\mathrm{B}\right)$,
and $\mathsf{Eff}\left(\mathrm{A}\right)$ become spanning sets of
suitable vector spaces over the real numbers, denoted by $\mathsf{St}_{\mathbb{R}}\left(\mathrm{A}\right)$,
$\mathsf{Transf}_{\mathbb{R}}\left(\mathrm{A},\mathrm{B}\right)$,
and $\mathsf{Eff}_{\mathbb{R}}\left(\mathrm{A}\right)$ respectively.
In this paper we will restrict our attention to finite systems, i.e.\ systems
$\mathrm{A}$ for which the vector spaces $\mathsf{St}_{\mathbb{R}}\left(\mathrm{A}\right)$
and $\mathsf{Eff}_{\mathbb{R}}\left(\mathrm{A}\right)$ are finite-dimensional.
Also, it will be assumed as a default that the sets $\mathsf{St}\left(\mathrm{A}\right)$,
$\mathsf{Transf}\left(\mathrm{A},\mathrm{B}\right)$, and $\mathsf{Eff}\left(\mathrm{A}\right)$
are compact in the topology induced by probabilities, by which one
has $\lim_{n\rightarrow+\infty}\mathcal{C}_{n}=\mathcal{C}$, where
$\mathcal{C}_{n},\mathcal{C}\in\mathsf{Transf}\left(\mathrm{A},\mathrm{B}\right)$,
if and only if 
\[
\lim_{n\rightarrow+\infty}\left(E\right|\mathcal{C}_{n}\otimes\mathcal{I}_{\mathrm{R}}\left|\rho\right)=\left(E\right|\mathcal{C}\otimes\mathcal{I}_{\mathrm{R}}\left|\rho\right)\quad\forall\mathrm{R},\forall\rho\in\mathsf{St}\left(\mathrm{A}\otimes\mathrm{R}\right),\forall E\in\mathsf{Eff}\left(\mathrm{B}\otimes\mathrm{R}\right).
\]

A \emph{test} from $\mathrm{A}$ to $\mathrm{B}$ is a collection
of transformations $\left\{ \mathcal{C}_{i}\right\} _{i\in\mathsf{X}}$
from $\mathrm{A}$ to $\mathrm{B}$, which can occur in an experiment
with outcomes in $\mathsf{X}$. If $\mathrm{A}$ (resp.\ $\mathrm{B}$)
is the trivial system, the test is called a \emph{preparation-test}
(resp.\ \emph{observation-test}). We stress that not all the collections
of transformations are tests: the specification of the collections
that are to be regarded as tests is part of the theory, the only requirement
being that the set of test is closed under parallel and sequential
composition.

If $\mathsf{X}$ contains a single outcome, we say that the test is
\emph{deterministic}. We will refer to deterministic transformations
as \emph{channels}. Following the most recent version of the formalism
\cite{Chiribella14}, we assume as part of the framework that every
test arises from an observation-test performed on one of the outputs
of a channel. The motivation for such an assumption is the idea that
the readout of the outcome could be interpreted physically as a measurement
allowed by the theory. Precisely, the assumption is the following.
\begin{assumption}[Physicalization of readout \cite{Chiribella14}]
\label{assu:read}For every pair of systems $\mathrm{A}$, $\mathrm{B}$,
and every test $\left\{ \mathcal{M}_{i}\right\} _{i\in\mathsf{X}}$
from $\mathrm{A}$ to $\mathrm{B}$, there exist a system $\mathrm{C}$,
a channel $\mathcal{M}\in\mathsf{Transf}\left(\mathrm{A},\mathrm{B}\otimes\mathrm{C}\right)$,
and an observation-test $\left\{ c_{i}\right\} _{i\in\mathsf{X}}\subset\mathsf{Eff}\left(\mathrm{C}\right)$
such that\[
\begin{aligned} \Qcircuit @C=1em @R=.7em @!R { & \qw \poloFantasmaCn{\rA}& \gate{\mathcal{M}_{i} } & \qw \poloFantasmaCn{\rB} &\qw }\end{aligned} ~=~\begin{aligned} \Qcircuit @C=1em @R=.7em @!R { & \qw \poloFantasmaCn{\rA}& \multigate{1}{\mathcal{M}} & \qw \poloFantasmaCn{\rB} & \qw \\ & & \pureghost{\mathcal M} & \qw \poloFantasmaCn{\rC} & \measureD{c_i}} \end{aligned} \qquad \forall i\in\mathsf{X}.
\]
\end{assumption}
A channel $\mathcal{U}$ from $\mathrm{A}$ to $\mathrm{B}$ is called
\emph{reversible} if there exists a channel $\mathcal{U}^{-1}$ from
$\mathrm{B}$ to $\mathrm{A}$ such that $\mathcal{U}^{-1}\mathcal{U}=\mathcal{I}_{\mathrm{A}}$
and $\mathcal{U}\mathcal{U}^{-1}=\mathcal{I}_{\mathrm{B}}$, where
$\mathcal{I}_{\mathrm{S}}$ is the identity channel on a generic system
$\mathrm{S}$. If there exists a reversible channel transforming $\mathrm{A}$
into $\mathrm{B}$, we say that $\mathrm{A}$ and $\mathrm{B}$ are
\emph{operationally equivalent}, denoted by $\mathrm{A}\simeq\mathrm{B}$.
The composition of systems is required to be \emph{symmetric}, meaning
that $\mathrm{A}\otimes\mathrm{B}\simeq\mathrm{B}\otimes\mathrm{A}$.

A state $\chi\in\mathsf{St}\left(\mathrm{A}\right)$ is called \emph{invariant}
if $\mathcal{U}\chi=\chi$, for every reversible channel $\mathcal{U}$.
Note that, in general, invariant states may not exist. In this paper
their existence will be a consequence of the axioms and of a standing
assumption of finite-dimensionality.

The pairing between states and effects leads naturally to a notion
of norm. We define the norm of a state $\rho$ as $\left\Vert \rho\right\Vert :=\sup_{a\in\mathsf{Eff}\left(\mathrm{A}\right)}\left(a|\rho\right)$.
The set of normalized (i.e.\ with unit norm) states of $\mathrm{A}$
will be denoted by $\mathsf{St}_{1}\left(\mathrm{A}\right)$. Similarly,
the norm of an effect $a$ is defined as $\left\Vert a\right\Vert :=\sup_{\rho\in\mathsf{St}\left(\mathrm{A}\right)}\left(a|\rho\right)$.
The set of normalized effects of system $\mathrm{A}$ will be denoted
by $\mathsf{Eff}_{1}\left(\mathrm{A}\right)$.

The probabilistic structure also offers an easy way to define pure
transformations. The definition is based on the notion of \emph{coarse-graining},
i.e.\ the operation of joining two or more outcomes of a test into
a single outcome. More precisely, a test $\left\{ \mathcal{C}_{i}\right\} _{i\in\mathsf{X}}$
is a \emph{coarse-graining} of the test $\left\{ \mathcal{D}_{j}\right\} _{j\in\mathsf{Y}}$
if there is a partition $\left\{ \mathsf{Y}_{i}\right\} _{i\in\mathsf{X}}$
of $\mathsf{Y}$ such that $\mathcal{C}_{i}=\sum_{j\in\mathsf{Y}_{i}}\mathcal{D}_{j}$
for every $i\in\mathsf{X}$. In this case, we say that $\left\{ \mathcal{D}_{j}\right\} _{j\in\mathsf{Y}}$
is a \emph{refinement} of $\left\{ \mathcal{C}_{i}\right\} _{i\in\mathsf{X}}$.
The refinement of a given transformation is defined via the refinement
of a test: if $\left\{ \mathcal{D}_{j}\right\} _{j\in\mathsf{Y}}$
is a refinement of $\left\{ \mathcal{C}_{i}\right\} _{i\in\mathsf{X}}$,
then the transformations $\left\{ \mathcal{D}_{j}\right\} _{j\in\mathsf{Y}_{i}}$
are a refinement of the transformation $\mathcal{C}_{i}$.

A transformation $\mathcal{C}\in\mathsf{Transf}(\mathrm{A},\mathrm{B})$
is called \emph{pure} if it has only trivial refinements, namely for
every refinement $\left\{ \mathcal{D}_{j}\right\} $ one has $\mathcal{D}_{j}=p_{j}\mathcal{C}$,
where $\left\{ p_{j}\right\} $ is a probability distribution. Pure
transformations are those for which the experimenter has maximal information
about the evolution of the system. We denote the set of pure transformations
from $\mathrm{A}$ to $\mathrm{B}$ as $\mathsf{PurTransf}\left(\mathrm{A},\mathrm{B}\right)$.
In the special case of states (resp.\ effects) of system $\mathrm{A}$
we use the notation $\mathsf{PurSt}\left(\mathrm{A}\right)$ (resp.\ $\mathsf{PurEff}\left(\mathrm{A}\right)$).
The set of normalized pure states (resp.\ effects) of $\mathrm{A}$
will be denoted by $\mathsf{PurSt}_{1}\left(\mathrm{A}\right)$ (resp.\ $\mathsf{PurEff}_{1}\left(\mathrm{A}\right)$).
As usual, non-pure states are called \emph{mixed}. 
\begin{defn}
Let $\rho$ be a normalized state. We say that a state $\sigma$ is
\emph{contained} in $\rho$ if we can write $\rho=p\sigma+\left(1-p\right)\tau$,
where $p\in\left(0,1\right]$ and $\tau$ is another state. 
\end{defn}
It is clear that no states are contained in a pure state, except the
pure state itself. At the opposite side there are \emph{completely
mixed} states \cite{Chiribella-informational}, such that every state
is contained in them.
\begin{defn}
\label{def:upon input}We say that two transformations $\mathcal{A},\mathcal{A}'\in\mathsf{Transf}\left(\mathrm{A},\mathrm{B}\right)$
are \emph{equal upon input} of the state $\rho\in\mathsf{St}_{1}\left(\mathrm{A}\right)$
if $\mathcal{A}\sigma=\mathcal{A}'\sigma$ for every state $\sigma$
contained in $\rho$. In this case we will write $\mathcal{A}=_{\rho}\mathcal{A}'$.
\end{defn}

\section{Axioms\label{sec:Axioms}}

Here we present our four axioms for diagonalizing states. As a first
axiom, we assume Causality, which forbids signalling from the future
to the past: 
\begin{ax}[Causality \cite{Chiribella-purification,Chiribella-informational}]
The outcome probabilities of a test do not depend on the choice of
other tests performed later in the circuit. 
\end{ax}
Causality is equivalent to the requirement that, for every system
$\mathrm{A}$, there exists a unique deterministic effect $u_{\mathrm{A}}$
on $\mathrm{A}$ (or simply $u$, when no ambiguity can arise). Thanks
to that, it is possible to define the \emph{marginal state} of a bipartite
state $\rho_{\mathrm{AB}}$ on system $\mathrm{A}$ as\[
\begin{aligned}\Qcircuit @C=1em @R=.7em @!R { & \prepareC{\rho_{\rA}}    & \qw \poloFantasmaCn{\rA} &  \qw   }\end{aligned}~=\!\!\!\!\begin{aligned}\Qcircuit @C=1em @R=.7em @!R { & \multiprepareC{1}{\rho_{\mathrm{AB}}}    & \qw \poloFantasmaCn{\rA} &  \qw   \\  & \pureghost{\rho_{\mathrm{AB}}}    & \qw \poloFantasmaCn{\rB}  &   \measureD{u} }\end{aligned}~.
\]In this case we will also write $\rho_{\mathrm{A}}:=\mathrm{Tr}_{\mathrm{B}}\rho_{\mathrm{AB}}$,
calling $u_{\mathrm{B}}$ as $\mathrm{Tr}_{\mathrm{B}}$, to remind
that the deterministic effect acts as the partial trace in quantum
theory. We will tend to keep the notation $\mathrm{Tr}$ in formulas
where the deterministic effect is directly applied to a state, e.g.\ $\mathrm{Tr}\:\rho:=\left(u|\rho\right)$.

In a causal theory (i.e.\ satisfying Causality), the norm of a state
$\rho$ is simply given by $\left\Vert \rho\right\Vert =\mathrm{Tr}\:\rho$.
Moreover, observation-tests are normalized in the following way (see
corollary 3 of Ref.~\cite{Chiribella-purification}):
\begin{prop}
\label{prop:characterization observation-tests}In a causal theory,
if $\left\{ a\right\} _{i\in\mathsf{X}}$ is an observation-test,
then $\sum_{i\in\mathsf{X}}a_{i}=u$.
\end{prop}
Causality guarantees that it is consistent to assume that the choice
of a test can depend on the outcomes of previous tests---namely that
it is possible to perform \emph{conditional tests} \cite{Chiribella-purification}.
Combined with the assumption of compactness, the ability to perform
conditional tests implies that every state is proportional to a normalized
state \cite{QuantumFromPrinciples}. Another consequence is that all
the sets $\mathsf{St}\left(\mathrm{A}\right)$, $\mathsf{Transf}\left(\mathrm{A},\mathrm{B}\right)$,
and $\mathsf{Eff}\left(\mathrm{A}\right)$ are \emph{convex}. In the
following we will take for granted the ability to perform conditional
tests, the fact that every state is proportional to a normalized state,
and the convexity of all the sets of transformations.

The second axiom in our list is Purity Preservation.
\begin{ax}[Purity Preservation\footnote{The name and the formulation of the axiom adopted here are the same
as in Ref.~\cite{Scandolo14}. The original axiom was called \emph{Atomicity
of Composition} \cite{D'Ariano} and involved \emph{only} sequential
composition. Extending the axiom to parallel composition is important
for our purposes, because it guarantees that the product of two pure
states is pure. In the axiomatization of Ref.~\cite{Chiribella-informational}
this property was a consequence of the Local Tomography axiom, which,
instead, is not assumed here.} \cite{D'Ariano,Chiribella-informational,Scandolo14,Chiribella-Scandolo15-1}]
Sequential and parallel compositions of pure transformations are
pure transformations. 
\end{ax}
We consider Purity Preservation as a fundamental requirement. Considering
the theory as an algorithm to make deductions about physical processes,
Purity Preservation ensures that, when presented with maximal information
about two processes, the algorithm outputs maximal information about
their composition \cite{Scandolo14}.

The third axiom is Purification. This axiom characterizes the physical
theories admitting a description where all deterministic processes
are pure and reversible at a fundamental level. Essentially, Purification
expresses a strengthened version of the principle of conservation
of information \cite{Chiribella-educational,Scandolo14}. In its simplest
form, Purification is phrased as a requirement about \emph{causal}
theories, where the marginal of a bipartite state is defined in a
canonical way. Specifically, we say that a state $\rho\in\mathsf{St}_{1}\left(\mathrm{A}\right)$
can be purified if there exists a pure state $\Psi\in\mathsf{PurSt}\left(\mathrm{A}\otimes\mathrm{B}\right)$
that has $\rho$ as its marginal on system $\mathrm{A}$. In this
case, we call $\Psi$ a \emph{purification} of $\rho$, and $\mathrm{B}$
a \emph{purifying system}. The axiom is as follows.
\begin{ax}[Purification \cite{Chiribella-purification,Chiribella-informational}]
Every state can be purified and two purifications with the same purifying
system differ by a reversible channel on the purifying system. 
\end{ax}
Technically, the second part of the axiom states that, if $\Psi,\Psi'\in\mathsf{PurSt}_{1}\left(\mathrm{A}\otimes\mathrm{B}\right)$
are such that $\mathrm{Tr}_{\mathrm{B}}\Psi_{\mathrm{AB}}=\mathrm{Tr}_{\mathrm{B}}\Psi'_{\mathrm{AB}}$,
then $\Psi'_{\mathrm{AB}}=\left(\mathcal{I}_{\mathrm{A}}\otimes\mathcal{U}_{\mathrm{B}}\right)\Psi_{\mathrm{AB}}$,
where $\mathcal{U}_{\mathrm{B}}$ is a reversible channel on $\mathrm{B}$.
In diagrams,\[
\begin{aligned}\Qcircuit @C=1em @R=.7em @!R { & \multiprepareC{1}{\Psi}    & \qw \poloFantasmaCn{\rA} &  \qw   \\  & \pureghost{\Psi}    & \qw \poloFantasmaCn{\rB}  &   \measureD{u} }\end{aligned}~=\!\!\!\! \begin{aligned}\Qcircuit @C=1em @R=.7em @!R { & \multiprepareC{1}{\Psi'}    & \qw \poloFantasmaCn{\rA} &  \qw   \\  & \pureghost{\Psi'}    & \qw \poloFantasmaCn{\rB}  &   \measureD{u} }\end{aligned} ~\Longrightarrow \begin{aligned}\Qcircuit @C=1em @R=.7em @!R { & \multiprepareC{1}{\Psi'}    & \qw \poloFantasmaCn{\rA} &  \qw   \\  & \pureghost{\Psi'}    & \qw \poloFantasmaCn{\rB}  &   \qw }\end{aligned}~=\!\!\!\! \begin{aligned}\Qcircuit @C=1em @R=.7em @!R { & \multiprepareC{1}{\Psi}    & \qw \poloFantasmaCn{\rA} &  \qw &\qw &\qw   \\  & \pureghost{\Psi}    & \qw \poloFantasmaCn{\rB}  &   \gate{\cU} &\qw \poloFantasmaCn{\rB} &\qw }\end{aligned}~.
\]

In quantum theory, the validity of Purification lies at the foundation
of all dilation theorems, such as Stinespring's \cite{Stinespring},
Naimark's \cite{Naimark}, and Ozawa's \cite{Ozawa}. In the finite-dimensional
setting, these theorems (or at least some aspects thereof) were reconstructed
axiomatically in \cite{Chiribella-purification}.

Finally, we introduce a new axiom, which we name \emph{Pure Sharpness}.
This axiom ensures that there exists at least one elementary property
associated with every system: 
\begin{ax}[Pure Sharpness]
 For every system $\mathrm{A}$, there exists at least one pure effect
$a\in\mathsf{PurEff}\left(\mathrm{A}\right)$ occurring with probability
1 on some state. 
\end{ax}
Pure Sharpness is reminiscent of the Sharpness axiom used in Hardy's
2011 axiomatization \cite{Hardy-informational-2}, which requires
a one-to-one correspondence between pure states and effects that distinguish
maximal sets of states.

\section{Consequences of the axioms\label{sec:Consequences-of-the}}

\subsection{\label{sub:Consequences-of-Causality,}Consequences of Causality,
Purity Preservation, and Purification}

Here we list a few consequences of the first three axioms, which will
become useful later.

The easiest consequence of Purification is that reversible channels
act transitively on the set of pure states (see lemma 20 in Ref.~\cite{Chiribella-purification}):
\begin{prop}
For any pair of pure states $\psi,\psi'\in\mathsf{PurSt}_{1}\left(\mathrm{A}\right)$
there exists a reversible channel $\mathcal{U}$ on $\mathrm{A}$
such that $\psi'=\mathcal{U}\psi$.
\end{prop}
As a consequence, every finite-dimensional system possesses one invariant
state (see corollary 34 of Ref.~\cite{Chiribella-purification}):
\begin{prop}
For every system $\mathrm{A}$, there exists a unique invariant state
$\chi_{\mathrm{A}}$, which is also a completely mixed state.
\end{prop}
Also, transitivity implies that the set of pure states is compact
for every system (see corollary 32 of Ref.~\cite{Chiribella-purification}).
This property is generally a non-trivial property---cf.\ Ref.~\cite{Entropy-Barnum}
for a counterexample of a state space with a non-closed set of pure
states.

A crucial consequence of Purification is the \emph{steering property}:
\begin{thm}[Steering property]
Let $\rho\in\mathsf{St}_{1}\left(\mathrm{A}\right)$ and let $\Psi\in\mathsf{PurSt}_{1}\left(\mathrm{A}\otimes\mathrm{B}\right)$
be a purification of $\rho$. Then $\sigma$ is contained in $\rho$
if and only if there exist an effect $b_{\sigma}$ on the purifying
system $\mathrm{B}$ and a non-zero probability $p$ such that\[
p\!\!\!\!\begin{aligned}\Qcircuit @C=1em @R=.7em @!R { & \prepareC{\sigma} & \qw \poloFantasmaCn{\rA} & \qw }\end{aligned}~=\!\!\!\!\begin{aligned}\Qcircuit @C=1em @R=.7em @!R { & \multiprepareC{1}{\Psi} & \qw \poloFantasmaCn{\rA} & \qw \\ & \pureghost{\Psi} & \qw \poloFantasmaCn{\rB} & \measureD{b_{\sigma}} }\end{aligned}~.
\]\end{thm}
\begin{proof}
The proof follows the same lines of theorem 6 and corollary 9 in Ref.~\cite{Chiribella-purification},
with the only difference that here we do not assume the existence
of perfectly distinguishable states. In its place, we use the framework
assumption~\ref{assu:read}, which guarantees that the outcome of
every test can be read out from a physical system.
\end{proof}
Now we introduce a definition and a proposition which will be used
later.
\begin{defn}
We say that a state $\rho\in\mathsf{St}_{1}\left(\mathrm{A}\otimes\mathrm{B}\right)$
is \emph{faithful for effects} of system $\mathrm{A}$ if, for any
$a,a'\in\mathsf{Eff}\left(\mathrm{A}\right)$, we have $a=a'$ if\[
\begin{aligned}\Qcircuit @C=1em @R=.7em @!R { & \multiprepareC{1}{\rho} & \qw \poloFantasmaCn{\rA} & \measureD{a} \\ & \pureghost{\rho} & \qw \poloFantasmaCn{\rB} & \qw }\end{aligned} ~= \!\!\!\! \begin{aligned}\Qcircuit @C=1em @R=.7em @!R { & \multiprepareC{1}{\rho} & \qw \poloFantasmaCn{\rA} & \measureD{a'} \\ & \pureghost{\rho} & \qw \poloFantasmaCn{\rB} & \qw }\end{aligned}
\]\end{defn}
\begin{prop}
\label{prop:faithful-effects}A pure state $\Psi_{\mathrm{AB}}$ is
faithful for effects of system $\mathrm{A}$ if and only if its marginal
$\omega_{\mathrm{A}}$ on $\mathrm{A}$ is completely mixed.
\end{prop}
See theorems 8 and 9 of Ref.~\cite{Chiribella-purification} for
the proof.

Combining Purification with Purity Preservation one obtains the following
properties:
\begin{prop}
\label{prop:characterization of effects}For every observation-test
$\left\{ a_{i}\right\} _{i\in\mathsf{X}}$ on $\mathrm{A}$, there
is a system $\mathrm{B}$ and a test $\left\{ \mathcal{A}_{i}\right\} _{i\in\mathsf{X}}\subset\mathsf{Transf}\left(\mathrm{A},\mathrm{B}\right)$
such that every $\mathcal{A}_{i}$ is pure and $a_{i}=u_{\mathrm{B}}\mathcal{A}_{i}$.
\end{prop}

\begin{prop}
\label{prop:non-disturbing measurement}Let $a$ be an effect such
that $\left(a|\rho\right)=1$, for some $\rho\in\mathsf{St}_{1}\left(\mathrm{A}\right)$.
Then there exists a transformation $\mathcal{T}$ on $\mathrm{A}$
such that $a=u\mathcal{T}$ and $\mathcal{T}=_{\rho}\mathcal{I}$,
where $\mathcal{I}$ is the identity.
\end{prop}
The proofs of the above propositions can be found in lemma 18 and
corollary 9 of Ref.~\cite{Chiribella-informational}.

Finally, thanks to Purification, proposition~\ref{prop:characterization observation-tests}
becomes also a sufficient conditions for a set of effects to be an
observation-test (cf.~theorem 18 of Ref.~\cite{Chiribella-purification}).
\begin{prop}
\label{prop:observation-test purification}A set of effects $\left\{ a_{i}\right\} _{i=1}^{n}$
is an observation-test if and only if $\sum_{i=1}^{n}a_{i}=u$.
\end{prop}

\subsection{Consequences of all the axioms}

In quantum theory, diagonalizing a state means decomposing it as a
convex combination of orthogonal pure states, i.e.\ pure states that
can be perfectly distinguished by a measurement.

In a general theory, perfectly distinguishable states are defined
as follows: 
\begin{defn}
The normalized states $\left\{ \rho_{i}\right\} $ are \emph{perfectly
distinguishable} if there exists an observation-test $\left\{ a_{j}\right\} $
such that $\left(a_{j}|\rho_{i}\right)=\delta_{ij}$. $\left\{ a_{j}\right\} $
is called \emph{perfectly distinguishing test}.
\end{defn}
Suppose we know that $\left(a|\rho\right)=1$, where $a$ is a pure
effect. Then, we can conclude that the state $\rho$ must be pure: 
\begin{prop}
\label{prop:uniqueness of state}Let $a\in\mathsf{PurEff}_{1}\left(\mathrm{A}\right)$.
Then, there exists a \emph{pure} state $\alpha\in\mathsf{PurSt}\left(\mathrm{A}\right)$
such that $\left(a|\alpha\right)=1$. Furthermore, for every $\rho\in\mathsf{St}\left(\mathrm{A}\right)$,
if $\left(a|\rho\right)=1$, then $\rho=\alpha$.
\end{prop}
See lemma 26 and theorem 7 of Ref.~\cite{Chiribella-informational}
for the proof idea.

Combining the above result with our Pure Sharpness axiom, we derive
the following 
\begin{prop}
 For every pure state $\alpha\in\mathsf{PurSt}\left(\mathrm{A}\right)$,
there exists at least one pure effect $a\in\mathsf{PurEff}\left(\mathrm{A}\right)$
such that $\left(a|\alpha\right)=1$.\end{prop}
\begin{proof}
By Pure Sharpness, there exists at least one pure effect $a_{0}$
such that $\left(a_{0}|\alpha_{0}\right)=1$ for some state $\alpha_{0}$.
By proposition~\ref{prop:uniqueness of state}, $\alpha_{0}$ is
pure. Now, for a generic pure state $\alpha$, by transitivity, there
is a reversible channel $\mathcal{U}$ such that $\alpha=\mathcal{U}\alpha_{0}$.
Hence, the effect $a:=a_{0}\mathcal{U}^{-1}$ is pure and $\left(a|\alpha\right)=1$.
\end{proof}
The above result will turn out to be useful for the construction of
our diagonalization procedure. A crucial ingredient in the derivation
of the diagonalization theorem is the following
\begin{thm}
\label{thm:diagonalization} Let $\rho$ be a normalized state of
system $\mathrm{A}$ and let $p_{*}$ be the probability defined as\footnote{Note that the maximum is well defined because the set of pure states
is compact, thanks to transitivity.} 
\[
p_{*}=\max_{\alpha\in\mathsf{PurSt}_{1}\left(\mathrm{A}\right)}\left\{ p\in\left[0,1\right]:\rho=p\alpha+\left(1-p\right)\sigma,\sigma\in\mathsf{St}_{1}\left(\mathrm{A}\right)\right\} .
\]
Let $\Psi\in\mathsf{PurSt}_{1}\left(\mathrm{A}\otimes\mathrm{B}\right)$
be a purification of $\rho$ and let $\widetilde{\rho}\in\mathsf{St}_{1}\left(\mathrm{B}\right)$
be the \emph{complementary state} of $\rho$, namely $\widetilde{\rho}:=\mathrm{Tr}_{\mathrm{A}}\Psi$.
Then, there exists a pure state $\beta\in\mathsf{PurSt}_{1}\left(\mathrm{B}\right)$
such that $\widetilde{\rho}=p_{*}\beta+\left(1-p_{*}\right)\tau$
for some state $\tau\in\mathsf{St}_{1}\left(\mathrm{B}\right)$.\end{thm}
\begin{proof}
By hypothesis, one can write $\rho=p_{*}\alpha+\left(1-p_{*}\right)\sigma$,
where $\alpha$ is a pure state and $\sigma$ is possibly mixed. Let
us purify $\rho$, and let $\Psi$ be one of its purifications, with
purifying system $\mathrm{B}$. According to the steering property,
there exists an effect $b$ that prepares $\alpha$ with probability
$p_{*}$, namely\begin{equation}\label{eq:steering diagonalization}
\begin{aligned}\Qcircuit @C=1em @R=.7em @!R { & \multiprepareC{1}{\Psi}    & \qw \poloFantasmaCn{\rA} &  \qw   \\  & \pureghost{\Psi}    & \qw \poloFantasmaCn{\rB}  &   \measureD{b} }\end{aligned}~=p_{*}\!\!\!\!\begin{aligned}\Qcircuit @C=1em @R=.7em @!R { & \prepareC{\alpha}    & \qw \poloFantasmaCn{\rA} &  \qw   }\end{aligned} ~.
\end{equation}Let $a$ be a pure effect such that $\left(a|\alpha\right)=1$. Applying
$a$ on both sides of Eq.~\eqref{eq:steering diagonalization}, we
get\[
\begin{aligned}\Qcircuit @C=1em @R=.7em @!R { & \multiprepareC{1}{\Psi}    & \qw \poloFantasmaCn{\rA} &  \measureD{a}   \\  & \pureghost{\Psi}    & \qw \poloFantasmaCn{\rB}  &   \measureD{b} }\end{aligned}~=p_{*}~.
\]On the other hand, applying $a$ to the state $\Psi$ we obtain\begin{equation}\label{eq:steering diagonalization B}
\begin{aligned}\Qcircuit @C=1em @R=.7em @!R { & \multiprepareC{1}{\Psi}    & \qw \poloFantasmaCn{\rA} &  \measureD{a}   \\  & \pureghost{\Psi}    & \qw \poloFantasmaCn{\rB}  &   \qw }\end{aligned}~=q\!\!\!\!\begin{aligned}\Qcircuit @C=1em @R=.7em @!R { & \prepareC{\beta}    & \qw \poloFantasmaCn{\rB} &  \qw   }\end{aligned}~,
\end{equation}where $q\in\left[0,1\right]$ and $\beta$ is a pure state (due to
Purity Preservation). Now if we apply $b$, we have\[
\begin{aligned}\Qcircuit @C=1em @R=.7em @!R { & \multiprepareC{1}{\Psi}    & \qw \poloFantasmaCn{\rA} &  \measureD{a}   \\  & \pureghost{\Psi}    & \qw \poloFantasmaCn{\rB}  &   \measureD{b} }\end{aligned}~=p_{*}~=q\!\!\!\!\begin{aligned}\Qcircuit @C=1em @R=.7em @!R { & \prepareC{\beta}    & \qw \poloFantasmaCn{\rB} &  \measureD{b}   }\end{aligned}~.
\]Since $\left(b|\beta\right)\in\left[0,1\right]$, we must have $q\geq p_{*}$.
We now prove that, in fact, equality holds. Let $\widetilde{b}$ be
a pure effect such that $\left(\widetilde{b}|\beta\right)=1$. Applying
$\widetilde{b}$ on both sides of Eq.~\eqref{eq:steering diagonalization B},
we obtain\[
\begin{aligned}\Qcircuit @C=1em @R=.7em @!R { & \multiprepareC{1}{\Psi}    & \qw \poloFantasmaCn{\rA} &  \measureD{a}   \\  & \pureghost{\Psi}    & \qw \poloFantasmaCn{\rB}  &   \measureD{\widetilde{b}} }\end{aligned}~=q.
\]By Purity Preservation, $\widetilde{b}$ will induce a pure state
on system $\mathrm{A}$, namely\[
q=\!\!\!\! \begin{aligned}\Qcircuit @C=1em @R=.7em @!R { & \multiprepareC{1}{\Psi}    & \qw \poloFantasmaCn{\rA} &  \measureD{a}   \\  & \pureghost{\Psi}    & \qw \poloFantasmaCn{\rB}  &   \measureD{\widetilde{b}} }\end{aligned}~=\widetilde{p}\!\!\!\!\begin{aligned}\Qcircuit @C=1em @R=.7em @!R { & \prepareC{\widetilde{\alpha}}    & \qw \poloFantasmaCn{\rA} &  \measureD{a}   }\end{aligned}~,
\]where $\widetilde{p}\in\left[0,1\right]$. From the above equation,
we have the inequality $q\leq\widetilde{p}$. Since by definition
we have $\widetilde{p}\leq p_{*}$, we finally get the chain of inequalities
$p_{*}\leq q\leq\widetilde{p}\leq p_{*}$, whence $p_{*}=q=\widetilde{p}$.
Hence, Eq.~\eqref{eq:steering diagonalization B} implies that the
pure state $\beta$ arises with probability $p_{*}$ in a convex decomposition
of the state $\widetilde{\rho}$. 
\end{proof}
A similar proof was used in lemma 30 of Ref.~\cite{Chiribella-informational}
in the special case where $\rho$ is the invariant state, and with
stronger assumptions, i.e.\ Ideal Compression, which is not assumed
here.

The effect $b$ that prepares $\alpha$ with probability $p_{*}$
can always be taken to be pure. Indeed, $\widetilde{b}$ is a \emph{pure}
effect that prepares the pure state $\widetilde{\alpha}$ on $\mathrm{A}$
with probability $\widetilde{p}$. But since $\widetilde{p}=p_{*}$,
then $\left(a|\widetilde{\alpha}\right)=1$. Therefore, by proposition~\ref{prop:uniqueness of state},
$\widetilde{\alpha}=\alpha$. This shows that $\alpha$ can always
be prepared with probability $p_{*}$ by using a pure effect on $\mathrm{B}$.

As a corollary we have the following:
\begin{cor}
\label{cor:equality p_*}Let $\rho\in\mathsf{St}_{1}\left(\mathrm{A}\right)$
be a state and let $\widetilde{\rho}\in\mathsf{St}_{1}\left(\mathrm{B}\right)$
be a complementary state of $\rho$. Let $p_{*}\left(\rho\right)$
and $p_{*}\left(\widetilde{\rho}\right)$ be defined like in theorem~\ref{thm:diagonalization},
for $\rho$ and $\widetilde{\rho}$ respectively. Then $p_{*}\left(\rho\right)=p_{*}\left(\widetilde{\rho}\right)$.\end{cor}
\begin{proof}
By theorem~\ref{thm:diagonalization}, we know that there exists
a pure state $\beta\in\mathsf{PurSt}_{1}\left(\mathrm{B}\right)$
arising in a convex decomposition of $\widetilde{\rho}$ with probability
$p_{*}\left(\rho\right)$:
\[
\widetilde{\rho}=p_{*}\left(\rho\right)\beta+\left(1-p_{*}\left(\rho\right)\right)\tau,
\]
where $\tau$ is another state of system $\mathrm{B}$. Therefore
$p_{*}\left(\widetilde{\rho}\right)\geq p_{*}\left(\rho\right)$.
By theorem~\ref{thm:diagonalization} applied to $\widetilde{\rho}$,
we know that there is a pure state $\alpha'\in\mathsf{PurSt}_{1}\left(\mathrm{A}\right)$
arising in a convex decomposition of $\rho$ with probability $p_{*}\left(\widetilde{\rho}\right)$:
\[
\rho=p_{*}\left(\widetilde{\rho}\right)\alpha'+\left(1-p_{*}\left(\widetilde{\rho}\right)\right)\sigma',
\]
where $\sigma'\in\mathsf{St}\left(\mathrm{A}\right)$. By definition
of $p_{*}\left(\rho\right)$, we have $p_{*}\left(\rho\right)\geq p_{*}\left(\widetilde{\rho}\right)$,
whence we conclude that $p_{*}\left(\rho\right)=p_{*}\left(\widetilde{\rho}\right)$.
\end{proof}
Now we are ready to prove the uniqueness of the pure effect associated
with a pure state. The proof uses the following lemma (see lemma 29
of Ref.~\cite{Chiribella-informational}).
\begin{lem}
\label{lem:invariant}Let $\chi$ be the invariant state of system
$\mathrm{A}$ and let $\alpha$ be a normalized pure state. Then 
\[
p_{\mathrm{max}}:=p_{\alpha}=\max\left\{ p:\exists\:\sigma,\chi=p\alpha+\left(1-p\right)\sigma\right\} 
\]
does not depend on $\alpha$.\end{lem}
\begin{prop}
For every normalized pure state $\alpha$ there exists a unique pure
effect $a$ such that $\left(a|\alpha\right)=1$.
\end{prop}
The proof is identical to the that of theorem 8 of Ref.~\cite{Chiribella-informational},
even though we are assuming fewer axioms.

We will denote by $\alpha^{\dagger}$ the unique pure effect associated
with the pure state $\alpha$, namely such that $\left(\alpha^{\dagger}|\alpha\right)=1$. 

We are able to establish a bijective correspondence between normalized
pure states and normalized pure effects. As a result, we obtain the
following corollary (cf.\ corollary 13 of Ref.~\cite{Chiribella-informational}).
\begin{cor}
For every pair of $a,a'\in\mathsf{PurEff}_{1}\left(\mathrm{A}\right)$,
there exists a reversible channel $\mathcal{U}$ on $\mathrm{A}$
such that $a'=a\mathcal{U}$.
\end{cor}

\section{Diagonalization of states\label{sec:Diagonalization-of-states}}

A \emph{diagonalization} of $\rho$ is a convex decomposition of $\rho$
into perfectly distinguishable pure states. The probabilities in such
a convex decomposition will be called the \emph{eigenvalues} of $\rho$.
Note that, since we are assuming the vector space $\mathsf{St}_{\mathbb{R}}\left(\mathrm{A}\right)$
to be finite-dimensional, diagonalizations of states will have a \emph{finite}
number of terms. Here we are not postulating the existence of perfectly
distinguishable pure states, but this will be a result of the present
set of axioms (see corollary~\ref{cor:existence pure perfectly}).

The starting point for diagonalization is the following 
\begin{prop}
\label{prop:effect diagonalization}Consider $\rho=p_{*}\alpha+\left(1-p_{*}\right)\sigma$,
where $p_{*}$ is defined in theorem~\ref{thm:diagonalization}.
We have $\left(\alpha^{\dagger}|\rho\right)=p_{*}$.\end{prop}
\begin{proof}
Let $\Psi_{\mathrm{AB}}$ be a purification of $\rho$. Then, the
proof of theorem~\ref{thm:diagonalization} yields the following
equality\[
\begin{aligned}\Qcircuit @C=1em @R=.7em @!R { & \multiprepareC{1}{\Psi}    & \qw \poloFantasmaCn{\rA} &  \measureD{\alpha^{\dagger}}   \\  & \pureghost{\Psi}    & \qw \poloFantasmaCn{\rB}  &   \qw }\end{aligned}~=p_{*}\!\!\!\!\begin{aligned}\Qcircuit @C=1em @R=.7em @!R { & \prepareC{\beta}    & \qw \poloFantasmaCn{\rB} &  \qw   }\end{aligned}
\]By applying the deterministic effect on both sides of the above equation,
we obtain\[
p_{*}=~p_{*}\!\!\!\!\begin{aligned}\Qcircuit @C=1em @R=.7em @!R { & \prepareC{\beta}    & \qw \poloFantasmaCn{\rB} &  \measureD{u}   }\end{aligned}~=\!\!\!\! \begin{aligned}\Qcircuit @C=1em @R=.7em @!R { & \multiprepareC{1}{\Psi}    & \qw \poloFantasmaCn{\rA} &  \measureD{\alpha^{\dagger}}   \\  & \pureghost{\Psi}    & \qw \poloFantasmaCn{\rB}  &  \measureD{u} }\end{aligned}~=\!\!\!\! \begin{aligned}\Qcircuit @C=1em @R=.7em @!R { & \prepareC{\rho}    & \qw \poloFantasmaCn{\rA} &  \measureD{\alpha^{\dagger}}   }\end{aligned}~.
\]This shows that $\left(\alpha^{\dagger}|\rho\right)=p_{*}$.
\end{proof}
The following proposition enables us to define $p_{*}$ in an alternative,
and perhaps simpler, way starting from measurements. 
\begin{prop}
Let $\rho\in\mathsf{St}_{1}\left(\mathrm{A}\right)$. Define $p^{*}:=\max_{a\in\mathsf{PurEff}_{1}\left(\mathrm{A}\right)}\left(a|\rho\right)$.
Then $p^{*}=p_{*}$.\end{prop}
\begin{proof}
By proposition~\ref{prop:effect diagonalization}, clearly one has
$p^{*}\geq p_{*}$. Since $p^{*}$ is the maximum, it is achieved
by some $a^{*}\in\mathsf{PurEff}_{1}\left(\mathrm{A}\right)$. Therefore,\[
p^{*} ~= \!\!\!\!\begin{aligned}\Qcircuit @C=1em @R=.7em @!R { & \prepareC{\rho}    & \qw \poloFantasmaCn{\rA} &  \measureD{a^{*}}   }\end{aligned}~=\!\!\!\!\begin{aligned}\Qcircuit @C=1em @R=.7em @!R { & \multiprepareC{1}{\Psi}    & \qw \poloFantasmaCn{\rA} &  \measureD{a^{*}}   \\  & \pureghost{\Psi}    & \qw \poloFantasmaCn{\rB}  &   \measureD{u} }\end{aligned}~,
\]where $\Psi_{\mathrm{AB}}$ is a purification of $\rho$. Now, $a^{*}$
prepares a pure state $\beta^{*}$ on $\mathrm{B}$ with probability
$q\leq p_{*}$ (cf.\ corollary~\ref{cor:equality p_*}).\[
p^{*} =\!\!\!\!\begin{aligned}\Qcircuit @C=1em @R=.7em @!R { & \multiprepareC{1}{\Psi}    & \qw \poloFantasmaCn{\rA} &  \measureD{a^{*}}   \\  & \pureghost{\Psi}    & \qw \poloFantasmaCn{\rB}  &   \measureD{u} }\end{aligned}~=q \!\!\!\!\begin{aligned}\Qcircuit @C=1em @R=.7em @!R { & \prepareC{\beta^{*}}    & \qw \poloFantasmaCn{\rB} &  \measureD{u}   }\end{aligned}~= q
\]We then obtain $p^{*}=q\leq p_{*}$, whence, in fact, $p^{*}=p_{*}$.
\end{proof}
The result expressed in proposition~\ref{prop:effect diagonalization}
has important consequences about diagonalization. Since $\left(\alpha^{\dagger}|\rho\right)=p_{*}$,
if $\rho=p_{*}\alpha+\left(1-p_{*}\right)\sigma$, then $\left(\alpha^{\dagger}|\sigma\right)=0$,
provided\footnote{If $p_{*}=1$, then $\rho$ is pure, and we are done. Therefore, without
loss of generality we can assume $p_{*}\neq1$.} $p_{*}\neq1$. Besides, if $\left(\alpha^{\dagger}|\sigma\right)=0$,
then $\left(\alpha^{\dagger}|\tau\right)=0$ for any state $\tau$
contained in $\sigma$. As a consequence, we have the following important
corollary, which guarantees the existence of perfectly distinguishable
pure states.
\begin{cor}
\label{cor:existence pure perfectly}Every pure state is perfectly
distinguishable from some other pure state.\end{cor}
\begin{proof}
Let us consider the invariant state $\chi$. For every normalized
pure state $\alpha$, we have $\chi=p_{\mathrm{max}}\alpha+\left(1-p_{\mathrm{max}}\right)\sigma$
(see lemma~\ref{lem:invariant}), where $\sigma$ is another normalized
state. By proposition~\ref{prop:effect diagonalization}, $\left(\alpha^{\dagger}|\sigma\right)=0$.
If $\sigma$ is pure, then $\alpha$ is perfectly distinguishable
from $\sigma$ by means of the observation-test $\left\{ \alpha^{\dagger},u-\alpha^{\dagger}\right\} $.
If $\sigma$ is mixed, than $\left(\alpha^{\dagger}|\psi\right)=0$
for every pure state $\psi$ contained in $\sigma$. Therefore $\alpha$
is perfectly distinguishable from $\psi$ again via the observation-test
$\left\{ \alpha^{\dagger},u-\alpha^{\dagger}\right\} $. 
\end{proof}
It is quite remarkable that the existence of perfectly distinguishable
(pure) states pops out from the axioms, without being assumed from
the start. In principle, the general theories considered in our framework
might not have had any perfectly distinguishable states at all!

\subsection{The diagonalization theorem}
\begin{thm}
\label{thm:diago}In a theory satisfying Causality, Purity Preservation,
Purification, and Pure Sharpness every state of every system can be
diagonalized.
\end{thm}
The proof uses the following lemma, which provides a condition for
the perfect distinguishability of a set of pure states:
\begin{lem}
\label{lem:distinguishable}If the pure states $\left\{ \alpha_{i}\right\} _{i=1}^{n}$
satisfy the condition $\left(\alpha_{i}^{\dagger}|\alpha_{j}\right)=0$
for every $j>i$, they are perfectly distinguishable.\end{lem}
\begin{proof}
By hypothesis, the observation-test $\left\{ \alpha_{i}^{\dagger},u-\alpha_{i}^{\dagger}\right\} $
distinguishes perfectly between $\alpha_{i}$ and all the other pure
states $\alpha_{j}$ with $j>i$. Equivalently, the test distinguishes
perfectly between $\alpha_{i}$ and the mixed state $\rho_{i}:=\frac{1}{n-i}\sum_{j>i}\alpha_{j}$.
As a result, we have the condition $\left(u-\alpha_{i}^{\dagger}|\rho_{i}\right)=1$.
Applying proposition~\ref{prop:non-disturbing measurement}, we can
construct a transformation $\mathcal{A}_{i}^{\perp}$, which occurs
with the same probability as $u-\alpha_{i}^{\dagger}$, such that
$\mathcal{A}_{i}^{\perp}=_{\rho_{i}}\mathcal{I}$, and, specifically,
\[
\mathcal{A}_{i}^{\perp}\alpha_{j}=\alpha_{j}\qquad\forall j>i.
\]
Moreover, the transformation $\mathcal{A}_{i}^{\perp}$ never occurs
on the state $\alpha_{i}$. Let $\left\{ \mathcal{A}_{i},\mathcal{A}_{i}^{\perp}\right\} $
be a binary test containing the transformation $\mathcal{A}_{i}^{\perp}$.
By construction, this test distinguishes without error between the
state $\alpha_{i}$ and all the states $\alpha_{j}$ with $j>i$,
in such a way that the latter are not disturbed. Using the tests $\left\{ \mathcal{A}_{i},\mathcal{A}_{i}^{\perp}\right\} $
it is easy to construct a protocol that distinguishes perfectly between
the states $\left\{ \alpha_{i}\right\} _{i=1}^{n}$. The protocol
works as follows: for $i$ going from $1$ to $n-1$, perform the
test $\left\{ \mathcal{A}_{i},\mathcal{A}_{i}^{\perp}\right\} $.
If the transformation $\mathcal{A}_{i}$ takes place, then the state
is $\alpha_{i}$. If the transformation $\mathcal{A}_{i}^{\perp}$
takes place, then perform the test $\left\{ \mathcal{A}_{i+1},\mathcal{A}_{i+1}^{\perp}\right\} $,
and so on.
\end{proof}

\begin{proof}[Proof of theorem~\ref{thm:diago}]
The proof consists of a constructive procedure for diagonalizing
arbitrary states. In order to diagonalize the state $\rho$, it is
enough to proceed along the following steps:
\begin{enumerate}
\item Set $\rho_{1}=\rho$ and $p_{*,0}=0$
\item Starting from $i=1$, decompose $\rho_{i}$ as $\rho_{i}=p_{*,i}\alpha_{i}+\left(1-p_{*,i}\right)\sigma_{i}$
as in theorem~\ref{thm:diagonalization}, and set $\rho_{i+1}=\sigma_{i}$,
$p_{i}=p_{*,i}\prod_{j=0}^{i-1}\left(1-p_{*,j}\right)$. If $p_{*,i}=1$,
then stop, otherwise continue to the step $i+1$.
\end{enumerate}
Recall that, at every step of the procedure, proposition~\ref{prop:effect diagonalization}
guarantees the condition $\left(\alpha_{i}^{\dagger}|\sigma_{i}\right)=0$.
Since by construction every state $\alpha_{j}$ with $j>i$ is contained
in the convex decomposition of $\sigma_{i}$, we also have $\left(\alpha_{k}^{\dagger}|\alpha_{l}\right)=0$
for $l>k$. Hence, lemma~\ref{lem:distinguishable} implies that
the states $\left\{ \alpha_{k}\right\} _{k=1}^{i}$, generated by
the first $i$ iterations of the protocol, are perfectly distinguishable,
for any $i$. For a finite dimensional system, this means that the
procedure has to terminate in a finite number of iterations. Once
the procedure has been completed, the state $\rho$ is diagonalized
as $\rho=\sum_{i}p_{i}\alpha_{i}$.
\end{proof}
Note that the diagonalization procedure in the above proof returns
a diagonalization of $\rho$ where the eigenvalues are naturally listed
in decreasing order, namely $p_{i}\geq p_{i+1}$ for every $i$. Such
an ordering will become useful when dealing with majorization.

\subsection{Unique vs non-unique diagonalization and the majorization criterion}

In quantum theory the diagonalization of every state is unique, up
to different choices of bases for degenerate eigenspaces. Is this
property satisfied by the operational diagonalization? In general,
it is conceivable that different diagonalization procedures may yield
different sets of eigenvalues for the same state. On top of that,
even our algorithm for diagonalizing states may not yield a single,
canonical diagonalization. It \emph{does} when the eigenvalues are
all distinct, but the situation may be different when two eigenvalues
coincide.

The uniqueness of the eigenvalues of a state is particularly important.
In a theory where the diagonalization is not unique any attempt to
define entropies from the eigenvalues is in serious danger of failure:
indeed, the resulting entropies would not be functions of the state,
but rather of its diagonalization. At this stage it is not clear whether
the present set of axioms (Causality, Purity Preservation, Purification,
and Pure Sharpness) implies that all the diagonalizations of a given
state have the same eigenvalues. We conjecture that the answer is
affirmative and plan to provide a rigorous proof in a forthcoming
paper \cite{Majorization}. For the moment, in this paper we will
prove an intermediate result, showing that the eigenvalues are unique
if one assumes the Strong Symmetry axiom by Barnum, Müller, and Ududec
\cite{Barnum-interference} in addition to our axioms.

\section{Combining diagonalization with Strong Symmetry\label{sec:Combining}}

Strong Symmetry is a requirement on the ability to transform maximal
sets of perfectly distinguishable pure states using reversible channels.
In general, a maximal set is defined as follows:
\begin{defn}
Let $\left\{ \rho_{i}\right\} _{i=1}^{n}$ be a set of perfectly distinguishable
states. We say that $\left\{ \rho_{i}\right\} _{i=1}^{n}$ is \emph{maximal}
if there is \emph{no} state $\rho_{n+1}$ such that the states $\left\{ \rho_{i}\right\} _{i=1}^{n+1}$
are perfectly distinguishable.
\end{defn}
When the maximal set is made of pure states, this definition gives
an operational characterization of the orthonormal bases of a finite-dimensional
Hilbert space. Another operational of characterization of them was
given in Ref.~\cite{Classical-structures} in terms of commutative
$\dagger$-Frobenius monoids.

With this definition, Strong Symmetry reads
\begin{ax}[Strong Symmetry \cite{Barnum-interference}]
The group of reversible channels acts transitively on maximal sets
of perfectly distinguishable \emph{pure} states.
\end{ax}
Strong Symmetry implies that all maximal sets of perfectly distinguishable
pure states have the same cardinality, sometimes referred to as the
\emph{dimension} of the system. We call a system of dimension $d$
a $d$\emph{-level system}. Note that in a $d$-level system, the
diagonalizations of a state have at most $d$ terms.

In the following we present a number of results arising from the combination
of diagonalization with Strong Symmetry. These results were preliminarily
discussed in the master's thesis of one of the authors \cite{Tesi}
and, more recently, they have appeared independently in Refs.~\cite{Krumm-thesis,Krumm-Muller}.
The first result is that the eigenvalues of the invariant state are
uniquely defined:
\begin{prop}
\label{prop:diagonalization chi d-level}Every diagonalization of
the invariant state $\chi=\sum_{i=1}^{d}p_{i}\alpha_{i}$ has $p_{i}=\frac{1}{d}$,
for every $i$.\end{prop}
\begin{proof}
Let $\left\{ a_{i}\right\} _{i=1}^{d}$ be the test that perfectly
distinguishes between the states $\left\{ \alpha_{i}\right\} _{i=1}^{d}$.
Then, one has $p_{i}=\left(a_{i}|\chi\right)$, for every $i$. Let
us consider all the possible permutations of the pure states $\left\{ \alpha_{i}\right\} _{i=1}^{d}$.
For instance, if $\pi\in S_{d}$, where $S_{d}$ is the symmetric
group over $d$ elements, we can consider the permuted states $\left\{ \alpha_{\pi\left(i\right)}\right\} _{i=1}^{d}$,
which are obviously still perfectly distinguishable. By Strong Symmetry,
there is a reversible channel $\mathcal{U}_{\pi}$ that implements
this permutation, namely $\alpha_{\pi\left(i\right)}=\mathcal{U}_{\pi}\alpha_{i}$.
Let us apply $\mathcal{U}_{\pi}$ to $\chi$.
\[
\chi=\sum_{j=1}^{d}p_{j}\mathcal{U}_{\pi}\alpha_{j}=\sum_{j=1}^{d}p_{j}\alpha_{\pi\left(j\right)}
\]
Now let us apply $a_{i}$ to $\chi$. We have
\[
p_{i}=\left(a_{i}|\chi\right)=\sum_{j=1}^{d}p_{j}\delta_{i,\pi\left(j\right)}=p_{\pi^{-1}\left(i\right)}.
\]
Since this holds for every $\pi\in S_{d}$, one has $p_{i}=p_{j}$
for every $j$. This implies that the eigenvalues are equal, therefore
$p_{i}=\frac{1}{d}$.
\end{proof}
Proposition~\ref{prop:diagonalization chi d-level} implies that
the pure states arising in every diagonalization of the invariant
state $\chi$ form a maximal set of perfectly distinguishable pure
states. One can wonder about the converse: is it true that every maximal
set of perfectly distinguishable pure states, combined with equal
weights, yields the invariant state? In this case, the answer is immediate
from Strong Symmetry:
\begin{prop}
\label{prop:diagonalization chi d-level 2}Let $\left\{ \psi_{i}\right\} _{i=1}^{d}$
be a maximal set of perfectly distinguishable pure states. Then one
has $\chi=\frac{1}{d}\sum_{i=1}^{d}\psi_{i}$.\end{prop}
\begin{proof}
Let us consider a diagonalization of $\chi$, say $\chi=\frac{1}{d}\sum_{i=1}^{d}\varphi_{i}$.
By Strong Symmetry, there is a reversible channel $\mathcal{U}$ such
that $\mathcal{U}\varphi_{i}=\psi_{i}$ for every $i$. Then we have
\[
\chi=\mathcal{U}\chi=\frac{1}{d}\sum_{i=1}^{d}\mathcal{U}\varphi_{i}=\frac{1}{d}\sum_{i=1}^{d}\psi_{i}.
\]

\end{proof}
So far we have used the diagonalization theorem as a ``black box'',
without referring to the axioms used to prove it. Using the full power
of the axioms allows us to prove stronger results. For example, we
are able to prove that every pure maximal set admits a pure, perfectly
distinguishing test:
\begin{lem}
\label{lem:pure test chi}For every pure maximal set $\left\{ \alpha_{i}\right\} _{i=1}^{d}$,
the pure effects $\left\{ \alpha_{i}^{\dagger}\right\} _{i=1}^{d}$
form an observation-test, which distinguishes perfectly between the
states $\left\{ \alpha_{i}\right\} _{i=1}^{d}$.\end{lem}
\begin{proof}
Let us consider the pure maximal set $\left\{ \alpha_{i}\right\} _{i=1}^{d}$.
By proposition~\ref{prop:diagonalization chi d-level 2}, we know
that $\chi=\frac{1}{d}\sum_{i=1}^{d}\alpha_{i}$. Let us prove that
the perfectly distinguishing test for $\left\{ \alpha_{i}\right\} _{i=1}^{d}$
is pure, namely made of the pure effects $\left\{ \alpha_{i}^{\dagger}\right\} _{i=1}^{d}$.
Recalling lemma~\ref{lem:invariant}, each $\alpha_{i}$ arises in
the the diagonalization of $\chi$ with weight $p_{*}=\frac{1}{d}$,
and one has that $\left(\alpha_{i}^{\dagger}|\alpha_{j}\right)=0$
for $j\neq i$. Therefore $\left(\alpha_{i}^{\dagger}|\alpha_{j}\right)=\delta_{ij}$.
We want now to prove that $\left\{ \alpha_{i}^{\dagger}\right\} _{i=1}^{d}$
is an observation-test. Thanks to Purification, it is sufficient to
show that $\sum_{i=1}^{d}\alpha_{i}^{\dagger}=u$ (see proposition~\ref{prop:observation-test purification}).
Let us consider a purification $\Phi_{\mathrm{AB}}$ of the invariant
state $\chi_{\mathrm{A}}$. By theorem~\ref{thm:diagonalization},
one has\[
\begin{aligned}\Qcircuit @C=1em @R=.7em @!R { & \multiprepareC{1}{\Phi}    & \qw \poloFantasmaCn{\rA} &  \measureD{\alpha^{\dagger}_{i}}   \\  & \pureghost{\Phi}    & \qw \poloFantasmaCn{\rB}  &  \qw }\end{aligned}~= \frac{1}{d}\!\!\!\! \begin{aligned}\Qcircuit @C=1em @R=.7em @!R { & \prepareC{\beta_{i}}    & \qw \poloFantasmaCn{\rB} &  \qw   }\end{aligned}~,
\]for some set of pure states $\left\{ \beta_{i}\right\} _{i=1}^{d}$.
By theorem~\ref{thm:diagonalization}, we know that a diagonalization
of the complementary state is $\widetilde{\rho}=\frac{1}{d}\sum_{i=1}^{d}\beta_{i}$.
Hence, we have the equality\[
\begin{aligned}\Qcircuit @C=1em @R=.7em @!R { & \multiprepareC{1}{\Phi}    & \qw \poloFantasmaCn{\rA} &  \measureD{\sum_{i=1}^{d}\alpha^{\dagger}_{i}}   \\  & \pureghost{\Phi}    & \qw \poloFantasmaCn{\rB}  &  \qw }\end{aligned}~= \!\!\!\! \begin{aligned}\Qcircuit @C=1em @R=.7em @!R { & \prepareC{\widetilde{\rho}}    & \qw \poloFantasmaCn{\rB} &  \qw   }\end{aligned}~= \!\!\!\! \begin{aligned}\Qcircuit @C=1em @R=.7em @!R { & \multiprepareC{1}{\Phi}    & \qw \poloFantasmaCn{\rA} &  \measureD{u}   \\  & \pureghost{\Phi}    & \qw \poloFantasmaCn{\rB}  &  \qw }\end{aligned}~,
\]where the last equality follows from the definition of the complementary
state $\widetilde{\rho}$. Since $\chi_{\mathrm{A}}$ is completely
mixed, $\Phi_{\mathrm{AB}}$ is faithful for effects of system $\mathrm{A}$
(by proposition~\ref{prop:faithful-effects}). Therefore we conclude
that $\sum_{i=1}^{d}\alpha_{i}^{\dagger}=u$, thus proving that $\left\{ \alpha_{i}^{\dagger}\right\} _{i=1}^{d}$
is an observation-test.
\end{proof}
The above lemma allows us to prove an important result, which will
be essential for the theory of majorization discussed in the next
section. The result is the following:
\begin{lem}
\label{lem:doubly stochastic}Let $\left\{ \psi_{i}\right\} _{i=1}^{d}$
and $\left\{ \varphi_{i}\right\} _{i=1}^{d}$ be two maximal sets
of perfectly distinguishable pure states. The matrix with entries
$\left(\psi_{i}^{\dagger}|\varphi_{j}\right)$ is doubly stochastic\footnote{See chapter 2, A.1 of Ref.~\cite{Olkin} for the definition of doubly
stochastic matrix.}.\end{lem}
\begin{proof}
Clearly $\left(\psi_{i}^{\dagger}|\varphi_{j}\right)\geq0$ because
$\left(\psi_{i}^{\dagger}|\varphi_{j}\right)$ is a probability. Let
us calculate $\sum_{i=1}^{d}\left(\psi_{i}^{\dagger}|\varphi_{j}\right)$.
By lemma~\ref{lem:pure test chi} we know that $\left\{ \psi_{i}^{\dagger}\right\} _{i=1}^{d}$
is an observation-test and therefore 
\[
\sum_{i=1}^{d}\left(\psi_{i}^{\dagger}|\varphi_{j}\right)=\left(u|\varphi_{j}\right)=1,
\]
because the $\varphi_{j}$'s are normalized. On the other hand, we
know that the invariant state can be decomposed as 
\[
\chi=\frac{1}{d}\sum_{j=1}^{d}\varphi_{j}=\frac{1}{d}\sum_{i=1}^{d}\psi_{i}
\]
(cf.\ proposition~\ref{prop:diagonalization chi d-level 2}). Hence,
we have
\[
\sum_{j=1}^{d}\left(\psi_{i}^{\dagger}|\varphi_{j}\right)=d\left(\psi_{i}^{\dagger}|\chi\right)=d\cdot\frac{1}{d}=1.
\]
This proves that the matrix with entries $\left(\psi_{i}^{\dagger}|\varphi_{j}\right)$
is doubly stochastic.
\end{proof}
Double stochasticity will be the key ingredient for the results of
the following section.

\section{Majorization and the resource theory of purity\label{sec:puriresource}}

Majorization is traditionally used as a criterion to compare the degree
of mixedness of probability distributions. Here we extend this approach
to general probabilistic theories satisfying our axioms and, provisionally,
Strong Symmetry. In order to define the degree of mixedness operationally,
we adopt the resource theory of purity defined in our earlier work
\cite{Chiribella-Scandolo15-1}, which considered the situation where
an experimenter has limited control on the dynamics of a closed system.
In this scenario, the set of free operations are the \emph{Random
Reversible (RaRe) channels}, defined as random mixtures of reversible
transformations:
\begin{defn}
A channel $\mathcal{R}$ is \emph{RaRe} if there exist a probability
distribution $\left\{ p_{i}\right\} _{i\in\mathsf{X}}$ and a set
of reversible channels $\left\{ \mathcal{U}_{i}\right\} _{i\in\mathsf{X}}$
such that $\mathcal{R}=\sum_{i\in\mathsf{X}}p_{i}\mathcal{U}_{i}$.
\end{defn}
By definition, RaRe channels cannot increase the purity of a state.
If $\rho=\mathcal{R}\sigma$, where $\mathcal{R}$ is a RaRe channel,
we say that $\rho$ is \emph{more mixed}\footnote{The same notion appeared in Ref.~\cite{Muller3D}, where it was used
to identify which states are better indicators of spatial directions.} than $\sigma$ \cite{Chiribella-Scandolo15-1}. If $\rho$ is more
mixed than $\sigma$ and $\sigma$ is more mixed than $\rho$ we say
that $\rho$ and $\sigma$ are \emph{equally mixed}.

Like in all resource theories, it is important to devise some methods
capable of detecting the convertibility of states under free operations
\cite{Spekkens}, which gives the (pre)ordering of states. We will
now show that, under the assumptions made so far in our paper, the
ordering of states according to their mixedness is completely determined
by majorization, just as it happens in quantum theory \cite{Nielsen-Chuang}.
Let us start by recalling the definition of majorization:
\begin{defn}
Let $\mathbf{x}$ and $\mathbf{y}$ be vectors in $\mathbb{R}^{d}$,
with the components arranged in \emph{decreasing} order. Then, $\mathbf{x}$
is \emph{majorized} by $\mathbf{y}$ (or $\mathbf{y}$ \emph{majorizes}
$\mathbf{x}$), and we write $\mathbf{x}\preceq\mathbf{y}$, if
\begin{itemize}
\item $\sum_{i=1}^{k}x_{i}\leq\sum_{i=1}^{k}y_{i}$, for every $k=1,\ldots,d-1$
\item $\sum_{i=1}^{d}x_{i}=\sum_{i=1}^{d}y_{i}$.
\end{itemize}
\end{defn}
It is known that $\mathbf{x}\preceq\mathbf{y}$ if and only if $\mathbf{x}=P\mathbf{y}$,
where $P$ is a doubly stochastic matrix \cite{Hardy-Littlewood-Polya1929,Olkin}.

Thanks to the results proved in the previous section, we are now in
the position to show that majorization of the eigenvalues is a necessary
condition for the mixedness ordering of two states:
\begin{thm}
\label{thm:mixedness -> majorization}In a theory satisfying Causality,
Purity Preservation, Purification, Pure Sharpness, and Strong Symmetry,
let $\rho$ and $\sigma$ be two states of a generic system and let
$\mathbf{p}$ and $\mathbf{q}$ be the vectors of the eigenvalues
in the diagonalizations of $\rho$ and $\sigma$. If $\rho$ is more
mixed than $\sigma$, then $\mathbf{p}\preceq\mathbf{q}$.\end{thm}
\begin{proof}
If $\rho$ is more mixed than $\sigma$, by definition, we have $\rho=\sum_{k}\lambda_{k}\mathcal{U}_{k}\sigma$,
where $\left\{ \lambda_{k}\right\} $ is a probability distribution
and $\mathcal{U}_{k}$ is a reversible channel, for every $k$. Suppose
$\rho=\sum_{j=1}^{d}p_{j}\psi_{j}$ and $\sigma=\sum_{j=1}^{d}q_{j}\varphi_{j}$
are diagonalizations of $\rho$ and $\sigma$. Then, $\rho=\sum_{k}\lambda_{k}\mathcal{U}_{k}\sigma$
becomes 
\[
\sum_{j=1}^{d}p_{j}\psi_{j}=\sum_{k}\lambda_{k}\sum_{j=1}^{d}q_{j}\mathcal{U}_{k}\varphi_{j}.
\]
By applying $\psi_{i}^{\dagger}$ we get 
\[
p_{i}=\sum_{j=1}^{d}q_{j}\sum_{k}\lambda_{k}\left(\psi_{i}^{\dagger}\right|\mathcal{U}_{k}\left|\varphi_{j}\right).
\]
This expression can be rewritten as $p_{i}=\sum_{j=1}^{d}P_{ij}q_{j}$,
where 
\[
P_{ij}:=\sum_{k}\lambda_{k}\left(\psi_{i}^{\dagger}\right|\mathcal{U}_{k}\left|\varphi_{j}\right)
\]
Now, $\left(\psi_{i}^{\dagger}\right|\mathcal{U}_{k}\left|\varphi_{j}\right)$
is a doubly stochastic matrix because $\left\{ \mathcal{U}_{k}\varphi_{j}\right\} _{j=1}^{d}$
is a maximal set of perfectly distinguishable pure states. Since the
set of doubly stochastic matrices is convex \cite{Olkin}, $P_{ij}$
is a doubly stochastic matrix, whence the thesis.
\end{proof}
As a corollary, we prove the desired result about the uniqueness of
the eigenvalues.
\begin{cor}
\label{cor:eigenvalues states}In a theory satisfying Causality, Purity
Preservation, Purification, Pure Sharpness, and Strong Symmetry, all
the diagonalizations of a given state have the same eigenvalues.\end{cor}
\begin{proof}
Let $\rho=\sum_{i=1}^{d}p_{i}\psi_{i}$ and $\rho=\sum_{j=1}^{d}q_{j}\varphi_{j}$
be two diagonalizations of a generic state $\rho$, and let $\mathbf{p}$
and $\mathbf{q}$ be the corresponding vectors of eigenvalues. Trivially,
$\rho$ is more mixed than $\rho$, which implies $\mathbf{p}\preceq\mathbf{q}$,
but also $\mathbf{q}\preceq\mathbf{p}$, therefore $\mathbf{p}=\Pi\mathbf{q}$,
for some permutation matrix $\Pi$ \cite{Olkin}. This means that
$\mathbf{p}$ and $\mathbf{q}$ differ only by a rearrangement of
their entries, whence the eigenvalues of $\rho$ are uniquely defined.
\end{proof}
In Ref.~\cite{Muller3D} Müller and Masanes proved that two states
that are equally mixed (in our terminology) differ by a reversible
channel. For theories satisfying the axioms adopted in this paper,
majorization provides an alternative proof:
\begin{prop}
In a theory satisfying Causality, Purity Preservation, Purification,
Pure Sharpness, and Strong Symmetry, two states $\rho$ and $\sigma$
are equally mixed under RaRe channels if and only if $\rho=\mathcal{U}\sigma$,
for some reversible channel $\mathcal{U}$. In particular, two equally
mixed states must have the same eigenvalues.\end{prop}
\begin{proof}
Sufficiency is straightforward. The proof of necessity is close to
the proof of corollary~\ref{cor:eigenvalues states}. If $\rho$
is equivalent to $\sigma$, then $\mathbf{p}\preceq\mathbf{q}$ and
$\mathbf{q}\preceq\mathbf{p}$, where $\mathbf{p}$ and $\mathbf{q}$
are the vectors of the eigenvalues of $\rho$ and $\sigma$ respectively.
This means that $\rho$ and $\sigma$ have the same eigenvalues (see
above). Thus, $\rho=\sum_{i=1}^{d}p_{i}\psi_{i}$ and $\sigma=\sum_{i=1}^{d}p_{i}\varphi_{i}$.
By Strong Symmetry, there exists a reversible channel $\mathcal{U}$
such that $\mathcal{U}\varphi_{i}=\psi_{i}$ for every $i$. Therefore,
\[
\mathcal{U}\sigma=\sum_{i=1}^{d}p_{i}\mathcal{U}\varphi_{i}=\sum_{i=1}^{d}p_{i}\psi_{i}=\rho.
\]

\end{proof}
We conclude this section by providing a complete equivalence between
majorization and the mixedness relation. While in theorem~\ref{thm:mixedness -> majorization}
we proved that majorization of the eigenvalues is a necessary condition
for the mixedness ordering, we now show that majorization is also
sufficient:
\begin{thm}
In a theory satisfying Causality, Purity Preservation, Purification,
Pure Sharpness, and Strong Symmetry, let $\rho$ and $\sigma$ be
two states of a generic system and let $\mathbf{p}$ and $\mathbf{q}$
be the vectors of their eigenvalues respectively. If $\mathbf{p}\preceq\mathbf{q}$,
then $\rho$ is more mixed than $\sigma$.\end{thm}
\begin{proof}
If $\mathbf{p}\preceq\mathbf{q}$, one has $\mathbf{p}=P\mathbf{q}$
for some doubly stochastic matrix $P$ \cite{Hardy-Littlewood-Polya1929,Olkin}.
Now, by Birkhoff's theorem \cite{Birkhoff,Olkin}, $P=\sum_{k}\lambda_{k}\Pi_{k}$,
where the $\Pi_{k}$'s are permutation matrices and $\left\{ \lambda_{k}\right\} $
is a probability distribution. Therefore $\mathbf{p}=\sum_{k}\lambda_{k}\Pi_{k}\mathbf{q}$;
specifically, this means that $p_{i}=\sum_{k}\lambda_{k}\sum_{j=1}^{d}\left[\Pi_{k}\right]_{ij}q_{j}$.
Therefore, we have 
\[
\rho=\sum_{i=1}^{d}p_{i}\psi_{i}=\sum_{i=1}^{d}\sum_{k}\lambda_{k}\sum_{j=1}^{d}\left[\Pi_{k}\right]_{ij}q_{j}\psi_{i}=
\]
\begin{equation}
=\sum_{k}\lambda_{k}\sum_{j=1}^{d}q_{j}\sum_{i=1}^{d}\left[\Pi_{k}\right]_{ij}\psi_{i}.\label{eq:tremilaelementidimatrice}
\end{equation}
Now, $\sum_{i=1}^{d}\left[\Pi_{k}\right]_{ij}\psi_{i}$ is a pure
state, given by $\psi_{\pi_{k}\left(i\right)}$ for a suitable permutation
$\pi_{k}\in S_{d}$. By Strong Symmetry, the permutation $\pi_{k}$
is implemented by a reversible channel $\mathcal{V}_{k}$. Moreover,
Strong Symmetry implies that there exists a reversible channel $\mathcal{U}$
such that $\mathcal{U}\varphi_{i}=\psi_{i}$ for every $i\in\left\{ 1,\ldots,d\right\} $.
Defining $\mathcal{U}_{k}:=\mathcal{V}_{k}\mathcal{U}$, we then have
\[
\mathcal{U}_{k}\varphi_{i}=\mathcal{V}_{k}\psi_{i}=\psi_{\pi_{k}(i)}=\sum_{i=1}^{d}\left[\Pi_{k}\right]_{ij}\psi_{i},
\]
which combined with Eq.~\eqref{eq:tremilaelementidimatrice} yields
\[
\rho=\sum_{k}\lambda_{k}\sum_{j=1}^{d}q_{j}\mathcal{U}_{k}\varphi_{j}=\sum_{k}\lambda_{k}\mathcal{U}_{k}\sigma.
\]
Hence, $\rho$ is more mixed than $\sigma$.
\end{proof}

\section{Conclusions\label{sec:Conclusions}}

In this work we have derived the diagonalization of states from four
basic operational axioms: Causality, Purity Preservation, Purification,
and Pure Sharpness. Our result has several applications: first of
all, it allows one to import all the known consequences of diagonalization
in the axiomatic context, such as those presented in Ref.~\cite{Barnum-interference},
where diagonalization was assumed as Axiom 1. For example, adding
Strong Symmetry, we obtain that the state space is self-dual---a property
that plays an important role in the reconstruction of quantum theory
\cite{Barnum-self}. The combination of our four axioms with Strong
Symmetry leads to important consequences, such as the fact that the
eigenvalues in the diagonalization of a state are uniquely determined.
While our results use Strong Symmetry, it remains as an open question
whether this requirement can be dropped or replaced by other, weaker
requirements. We conjecture that this is indeed the case, and we plan
to investigate the issue further in a forthcoming paper \cite{Majorization}.

Another important application of our results is in the axiomatic reconstruction
of (quantum) thermodynamics. In a previous work \cite{Chiribella-Scandolo15-1},
we defined an operational resource theory of purity---dual to the
resource theory of entanglement---in which free operations are random
reversible channels. A natural application of the diagonalization
theorem is the formulation of a majorization criterion capable of
detecting whether a thermodynamic transition is possible or not, and
to establish quantitative measures of mixedness \cite{Tesi}. Specifically,
when Strong Symmetry is added to our axioms, the ordering of states
in the operational resource theory of purity is completely characterized
by the majorization criterion. Such an application contributes also
to the difficult problem of finding the right requirements that guarantee
a well-behaved notion of entropy in general probabilistic theories
\cite{Entropy-Barnum,Entropy-Short,Entropy-Kimura}. To some extent,
our results suggest that having a sensible notion of entropy (and
therefore having a sensible thermodynamics) is not a generic feature
of general probabilistic theories, but rather a quite stringent constraint.
In addition to the application to the axiomatization of quantum thermodynamics,
it is our hope that this work will contribute to the development of
an axiomatic approach to information theory---in particular including
data compression and transmission over noisy channels.

\paragraph{Acknowledgements}

We acknowledge P Perinotti for a useful discussion on the fermionic
quantum theory of Refs.~\cite{Fermionic1,Fermionic2}. This work
is supported by Foundational Questions Institute through the large
grant ``The fundamental principles of information dynamics'' (FQXi-RFP3-1325),
by the National Natural Science Foundation of China through Grants
11450110096 and 11350110207, and by the 1000 Youth Fellowship Program
of China. The research by CMS has been supported by a scholarship
from ``Fondazione Ing.\ Aldo Gini'' and by the Chinese Government
Scholarship. 

\bibliographystyle{eptcs}
\bibliography{bibliographyQPL}

\end{document}

%% file: myQcircuit.tex
\usepackage[matrix,frame,arrow]{xy}
\usepackage{amsmath}

\newcommand{\qw}[1][-1]{\ar @{-} [0,#1]}

\newcommand{\gate}[1]{*{\xy *+<.6em>{#1};p\save+LU;+RU **\dir{-}\restore\save+RU;+RD **\dir{-}\restore\save+RD;+LD **\dir{-}\restore\POS+LD;+LU **\dir{-}\endxy} \qw}

\newcommand{\measureD}[1]{*{\xy*+=+<.5em>{\vphantom{\rule{0em}{.1em}#1}}*\cir{r_l};p\save*!R{#1} \restore\save+UC;+UC-<.5em,0em>*!R{\hphantom{#1}}+L **\dir{-} \restore\save+DC;+DC-<.5em,0em>*!R{\hphantom{#1}}+L **\dir{-} \restore\POS+UC-<.5em,0em>*!R{\hphantom{#1}}+L;+DC-<.5em,0em>*!R{\hphantom{#1}}+L **\dir{-} \endxy} \qw}

\newcommand{\multigate}[2]{*+<1em,.9em>{\hphantom{#2}} \qw \POS[0,0].[#1,0];p !C *{#2},p \save+LU;+RU **\dir{-}\restore\save+RU;+RD **\dir{-}\restore\save+RD;+LD **\dir{-}\restore\save+LD;+LU **\dir{-}\restore}
\newcommand{\Qcircuit}[1][0em]{\xymatrix @*=<#1>} 

\newcommand{\pureghost}[1]{*+<1em,.9em>{\hphantom{#1}}}
\newcommand{\multiprepareC}[2]{*+<1em,.9em>{\hphantom{#2}}\save[0,0].[#1,0];p\save !C
  *{#2},p+RU+<0em,0em>;+LU+<+.8em,0em> **\dir{-}\restore\save +RD;+RU **\dir{-}\restore\save
  +RD;+LD+<.8em,0em> **\dir{-} \restore\save +LD+<0em,.8em>;+LU-<0em,.8em> **\dir{-} \restore \POS
  !UL*!UL{\cir<.9em>{u_r}};!DL*!DL{\cir<.9em>{l_u}}\restore}
\newcommand{\prepareC}[1]{*{\xy*+=+<.5em>{\vphantom{#1\rule{0em}{.1em}}}*\cir{l^r};p\save*!L{#1} \restore\save+UC;+UC+<.5em,0em>*!L{\hphantom{#1}}+R **\dir{-} \restore\save+DC;+DC+<.5em,0em>*!L{\hphantom{#1}}+R **\dir{-} \restore\POS+UC+<.5em,0em>*!L{\hphantom{#1}}+R;+DC+<.5em,0em>*!L{\hphantom{#1}}+R **\dir{-} \endxy}}
\newcommand{\poloFantasmaCn}[1]{{{}^{#1}_{\phantom{#1}}}}